%% file: regret.tex
\documentclass[10pt,conference]{IEEEtran}
\usepackage{amssymb}
\usepackage{amsmath}
\usepackage{graphicx}
\usepackage{bm}
\usepackage{bbm}
\usepackage{subfigure}
\usepackage{psfrag}
\usepackage{pstricks, pstricks-add}
\usepackage{colortbl}
\usepackage{mathrsfs}
\usepackage{setspace}
\usepackage[active]{srcltx}
\title{On Regret of Parametric Mismatch in  Minimum Mean Square Error Estimation}
\author{Majid Fozunbal\\
        Hewlett-Packard Laboratories\\
        Palo Alto, CA 94304}
\input{definitions.tex}

\input{environments.tex}

\input{localmacro.tex}

\newcommand{\figfontsmall}{.7}
\newcommand{\figfontlarge}{.9}

\begin{document}
\maketitle
\begin{abstract}
    This paper studies the effect of parametric mismatch  in  minimum mean
    square error (MMSE) estimation. In particular, we consider the problem of estimating the input signal
    from the output of an additive white Gaussian channel whose gain is fixed, but unknown.
    The input distribution is known, and the estimation process consists of two algorithms. First, a channel
    estimator blindly estimates the channel gain using past observations. Second, a mismatched
    MMSE estimator, optimized for the estimated channel gain, estimates the input signal.
    We analyze the {\em regret}, i.e., the additional mean square error, that is raised in this process.
    We derive upper-bounds on both absolute and relative regrets. Bounds are expressed in terms of
    the Fisher information. We also study regret for unbiased, efficient channel estimators, and derive
    a simple trade-off between Fisher information and relative regret.
    This trade-off shows that the product of a certain function of
    relative regret and Fisher information equals  the signal-to-noise ratio, independent of
    the input distribution.
    The trade-off  relation implies
    that higher Fisher information results to smaller~expected~relative~regret.
\end{abstract}

    \section{Introduction}
        Consider an application that you are given the output of a system, and you seek
        to recover the input of the system. You know that the system is noisy, e.g., it adds white Gaussian noise
        to the output. You know the
        distribution of the input, but you do not know the system parameters. Problems
        of this sort arise in different applications in signal processing and communication systems.
        Some examples include blind deconvolution \cite{Donoho:Deconvolution}, dereverberation \cite{Majid:ICASSP:10},
        denoising \cite{Weissman:DUDE:Uncetainty}, and mismatch decoding \cite{Mismatch:Lapidoth}.
        These applications differ in their fundamental models, fidelity criteria,
        and methodologies.  However, they have one thing in common: they all suffer from parametric mismatch in
        recovering the input signals.

        The motivation of this work is blind deconvolution and dereverberation applications. Linear time-invariant
        channels serve as common models in these applications. As the input signal passes through these channels, it convolves
        with the unknown finite-impulse response (FIR) of the channel, and it adds with additive white Gaussian noise (of known variance).
        Recovering the input signals from the noisy output could be impossible even with perfect knowledge about the
        channel response. This is out of our scope. Instead, we aim to study the penalty and performance degradation that
        is specifically caused by the lack of knowledge about the channel response.

        We benchmark performance
        against that of perfect channel knowledge scenario. We are concerned about
        issues such as required sample complexity or training in channel estimation to bring performance of input estimation
        within a desired range. As a counterpart problem in communication systems, one may think of block fading channels
        and the trade-off between accuracy of channel estimation and performance of decoding~\cite{Hassibi:Training}. Note
        that channel estimation in our case is blind as we have no control of the source.

        As a first step to address these problems, in this work, we focus on the most basic system in which
        the unknown channel is just a single gain.  We expect that the results and intuitions
        of this work will shed lights on the analysis of generic FIR channels.\footnote{Analogous to the case between the analysis of flat-fading and the analysis of frequency-selective channels.}
        In treating the problem, we consider an estimation process that consists of two algorithms. First, a {\em channel
        estimator} blindly estimates the channel gain using past output observations. Second, a mismatched
        {\em minimum mean square error} (MMSE) estimator, optimized for the estimated channel gain, estimates the input signal.
        Figure~\ref{fig:system} illustrates the building blocks of this process. Due to estimation error in
        channel estimation, the MMSE estimator that is used in recovering the input signal results in a
        mean square error that is larger than that of the ideal MMSE estimator.
        We call this additional error as {\em regret}, and
        we derive novel upper-bounds on both absolute and relative regrets. The bounds are simple and
        demonstrate interesting connections to  the Fisher information. To this end, one might attempt
        to exploit the results of \cite{Sergio:Mismatched} and
        \cite{Guo:RelativeEntropy:Score} to derive other alternative bounds.

        We also quantify  regret for unbiased, {\em efficient} channel estimators. Since these
        estimators achieve Cramer-Rao bound, they result in  a simple trade-off relation between Fisher information and relative regret.
        This trade-off relation expresses that the product of a certain function of
        relative regret and Fisher information is equivalent to the signal-to-noise ratio, independent of the input distribution.
        Trade-off suggests that higher Fisher information results to smaller expected relative regret.
        Although, intuitively, this may seem expected, simplicity of the trade-off relation makes it worthwhile.

    \section{Setup}
        Consider a linear dynamic system
        \begin{eqnarray}
            Y_n = aX_{n}+ V_n
            \label{eqn:dynamic}
        \end{eqnarray}
        in which $\{V_n\}$ is an independent, identically, distributed (i.i.d.) Gaussian noise
        such that $V_n\sim\mcal N(0,\s^2_v)$. The input $X_{n}$
        is an i.i.d. process whose distribution is known to be $\probdist{}{X}$.
        Parameter $a \in \mbb R^+$ is a fixed,  unknown channel gain. It results to
        a derived parametric family of probability measures $\probdist{a}{X,Y}$,
        the joint distribution of $X$ and $Y$, governing the system dynamic \eqref{eqn:dynamic}.
        The objective is to observe a realization of the output process
        \[Y^n=(Y_1,Y_2, \cdots, Y_n)\] and estimate the realization of the
        underlying input process, i.e.,
        \[X^n=(X_1, X_2, \cdots, X_n).\]

        Let $\mcal X =\mbb R$ and $\mcal Y = \mbb R$ denote the input and output spaces, respectively.
        We consider memoryless input estimators, e.g., $\phi \colon \mcal Y \to \mcal X$ where
        $\phi(Y_n)$ is an estimate for $X_n$.  The {\em mean square error} (MSE) for  $\f$ is defined
        \begin{eqnarray}
             \Exp{(X-\phi(Y))^2} = \int (x-\phi(y))^2 d \probdistn{a}.
            \label{eqn:mse}
        \end{eqnarray}
        In Eq. \eqref{eqn:mse} and henceforth we follow the convention that
        unsubscribed expectations are measured according to $\probdist{a}{X,Y}$.
        Moreover, we use concise notations like $\probdistn{a}=\probdist{a}{X,Y}$
        and $\probdistn{a|y}=\probdist{a}{X|Y=y}$ to denote joint and conditional
        distributions, respectively.

        One seeks to find an estimator that minimizes MSE~\eqref{eqn:mse}. The main challenge, however, is that $a$ and
        $\probdistn{a}$ are unknown. If we had
        oracle knowledge about $a$, the MMSE estimator for $X$
        is defined
        \begin{align}
             \phi_a(y) = \Exp{X|Y=y}.
             \label{eqn:optimal:mmse:estimator}
        \end{align}
        for an observation $Y=y$. Any other estimator $\f$ results to additional
        error that we call it {\em regret}. The motivation for this name is that
        it measures degradation on performance, an impact caused by imprecise knowledge about $a$.

        In this paper, we assess regret for a special class of mismatched estimators. Namely,
        we consider an estimation process that is depicted in Figure \ref{fig:system}. A channel
        estimation works in parallel with an MMSE input estimation as follows. At time instance $n$,
        a channel estimator  finds an  estimate $\hat{a}=\hat{a}_n$ of $a$ using the observed values $Y^{n-1}$.
        Then, it uses the optimal estimator of $\probdist{\hat{a}}{X, Y}$
        to compute
        \begin{align}
            \phi_{\hat{a}}(y_n) = \Expi{\hat{a}}{X_n|Y=y_n}
            \label{eqn:mismatch:estimator}
        \end{align}
        as an estimate for $X_n$. Function  $\phi_{\hat{a}}$ is a mismatch MMSE
        estimator that causes regret when used in place of $\f_a$. In the following sections, we
        study two types of regret: {\em absolute regret} and {\em relative regret}.

            \begin{figure}[t]
                \begin{center}
                    \psfrag{X}[Bc][B1][\figfontsmall][0]{$X_n$}
                    \psfrag{Y}[Bc][B1][\figfontsmall][0]{$Y_n$}%=aX_n+V_n$}
                    \psfrag{Yn}[Bc][B1][\figfontsmall][0]{$Y^{n-1}$}%=aX_n+V_n$}
                    \psfrag{V}[Bc][B1][\figfontsmall][0]{$V_n$}
                    \psfrag{a}[Bc][B1][\figfontsmall][0]{$a$}
                    \psfrag{aX}[Bc][B1][\figfontsmall][0]{$aX_n$}
                    \psfrag{ah}[Bc][B1][\figfontsmall][0]{$\hat{a}$}
                    \psfrag{Xh}[Bc][B1][\figfontsmall][0]{$\hat{X}_n$}
                    \psfrag{C}[Bc][B1][\figfontsmall][0]{channel}
                    \psfrag{E}[Bc][B1][\figfontsmall][0]{estimation}
                    \psfrag{I}[Bc][B1][\figfontsmall][0]{MMSE}
                    \psfrag{e}[Bc][B1][\figfontsmall][0]{estimation process}
                    \includegraphics[width=.95\columnwidth]{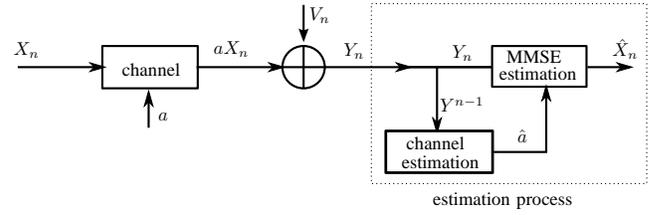}
                \end{center}
                \caption{Figure depicts the building blocks of the system setup.
                The estimation process consists of two individual algorithms: 1) a channel estimation algorithm
                that blindly estimates $\hat{a}$ as an estimate for $a$, 2) a mismatch MMSE estimation algorithm, optimized for $\hat{a}$, that recovers $X_n$.}
                \label{fig:system}
            \end{figure}

    \section{Absolute Regret}
        \subsection{Deviation Analysis}
        The absolute regret corresponding to $\phi_{\hat{a}}$ is
        \begin{align}
            \regret{\hat{a},a} = \Exp{(X-\phi_{\hat{a}}(Y))^2} - \Exp{(X-\phi_a(Y))^2}.
        \end{align}
        Application of orthogonality principle results to
        \begin{align}
            \regret{\hat{a},a} = \Exp{(\phi_{\hat{a}}(Y)  - \phi_a(Y))^2}.
            \label{eqn:regret}
        \end{align}
        Eq. \eqref{eqn:regret} quantifies the absolute {\em regret} of using $\phi_{\hat{a}}$ instead of $\phi_a$. The following lemma
        states and proves an upper-bound on  \eqref{eqn:regret}.
        \begin{lemma}
            For every $\hat{a}$, the following holds true
            \begin{align}
                 \nonumber \regret{\hat{a},a}  \leq (\hat{a}-a)^2 & \Exp{\left(6 \s^2_x + 8 \frac{Y^2}{{a}^2}\right)  \fisher{X;a||Y}}\\&+o(\hat{a}-a)^2
                \label{eqn:lemma:upperbound:regret}
            \end{align}
            in which the expectation is with respect to $Y$,
            %\[k(Y) \triangleq 6 \s^2_x + 8 \frac{Y^2}{{a}^2}+4 \frac{Y^2}{{\hat{a}}^2},\]
            and
            \begin{eqnarray}
                 \fisher{X;a||Y}  \triangleq \Expi{}{\big(\nabla \ln \dfi{a}{X|Y} \big)^2 |Y}
                 \label{eqn:lemma:fisher:random}
            \end{eqnarray}
            is the Fisher information of $X$ relative to $a$, conditioned on $Y$.
            Here, $\dfi{a}{X|Y}$ is  the density of $\probdistn{{a}|Y}$.
            \label{lemma:upperbound:regret}
        \end{lemma}
        \begin{proof}
            Refer to Appendix \ref{app:upperbound:regret}.
        \end{proof}
        Lemma \ref{lemma:upperbound:regret} describes a bound \eqref{eqn:lemma:upperbound:regret} that
        comprises two  multiplicative terms. The first term $(\hat{a}-a)^2$ measures the
        channel estimation error. The second term is the weighted average of conditional
        Fisher information. Intuitively, this term  measures the amount of information
        that an observable random variable $X$ carries about unknown parameter $a$ conditioned on $Y$, assigning
        more weight to larger values of $Y$.

        \begin{corollary}
            If  $|\hat{a}-a| << 1$, and
%            \begin{align}
%                \regret{\hat{a},a} \leq (\hat{a}-a)^2\Exp{k(Y) \fisher{X;a|Y=y}}.
%                \label{eqn:upperbound:regret:small:deviation}
%            \end{align}
%            Moreover,
            if $\fisher{X;a||Y}$ and $Y^2$ are uncorrelated, we obtain the simple
            bound
            \begin{align}
                \regret{\hat{a},a} \leq (\hat{a}-a)^2 (14 \s^2_x + 8 \frac{\s^2_v}{{a}^2}) \fisher{X;a|Y}
                \label{eqn:upperbound:regret:small:deviation:simple}
            \end{align}
            in which
                \begin{eqnarray}
                    \fisher{X;a|Y} \triangleq \Expi{}{\big(\nabla \ln \dfi{a}{X|Y} \big)^2}
                    \label{eqn:fisher:average}
                \end{eqnarray}
            is the average of ${\fisher{X;a||Y}}$ with respect to $Y$.\footnote{Lookout for the subtle notational difference between ${\fisher{X;a||Y}}$, a random variable,
            and $\fisher{X;a|Y}$, a scalar.}
        \label{cor:regret:approximation}
        \end{corollary}

        \subsection{Efficient Channel Estimation}
        Neither Eq. \eqref{eqn:lemma:upperbound:regret} nor Eq. \eqref{eqn:upperbound:regret:small:deviation:simple} depend on the channel estimation
        algorithm that estimates $a$. They simply relate small deviation between $\hat{a}$ and $a$ to
        absolute regret in estimating $X$. To incorporate the effect of channel estimation algorithm, we proceed as follows.

        As mentioned earlier, at time $n$, $\hat{a}$ is obtained through
        observation of $Y^{n-1}=(Y_i)_{i=1}^{n-1}$. In formal terms,
        \[\hat{a} = A_n(Y^{n-1})\]
        where $A=(A_1, A_2, \cdots)$ is a channel estimation algorithm in which $A_n\colon \mcal Y^{n-1} \to \mbb R^+$.
        \begin{lemma}
            Let $\mcal A$ denote the class of all unbiased channel estimation algorithms.
            If $\mcal A$ contains an {\em efficient estimator} \cite[p. 92]{Borovkov:MathematicalStatistics},
            the following holds true
            \begin{align}
                \nonumber \inf_{A\in \mcal A} \Exp{\regret{A_n(Y^{n-1}),  a}} \leq \qquad \qquad \qquad\\ \frac{1}{n-1}  \frac{\Exp{\left(6 \s^2_x + 8 \frac{Y^2}{{a}^2}\right)  \fisher{X;a||Y}}}{\fisher{Y;a}}
                \label{eqn:esterror:aregret:Cramer-Rao}
            \end{align}
            for sufficiently large values of $n$.\footnote{The expectation in the LHS is with respect to $Y^{n-1}$.}
            \label{lemma:regret:Cramer-Rao}
        \end{lemma}
        \begin{proof}
            Refer to Appendix \ref{app:prf:regret:Cramer-Rao}.
        \end{proof}

    \section{Relative Regret}
        \subsection{Deviation Analysis}
        Let
        \begin{align}
            \rregret{\hat{a},a} = \Exp{\frac{(\phi_{\hat{a}}(Y)  - \phi_a(Y))^2}{\Expi{\hat{a}}{X^2|Y} + \Expi{{a}}{X^2|Y}}}
            \label{eqn:relative:regret}
        \end{align}
        denote the relative regret. The following lemma states and proves a simple upper-bound
        on  $ \rregret{\hat{a},a}$.
        \begin{lemma}
            For every $\hat{a}$, we have
            \begin{align}
                \rregret{\hat{a},a}  \leq & {(\hat{a}-a)^2} \;  \fisher{X;a|Y} + o(\hat{a}-a)^2
                \label{eqn:lemma:relative:regret}
            \end{align}
            where $\fisher{X;a|Y}$ is defined as Eq. \eqref{eqn:fisher:average}
            and denotes the conditional Fisher information of $X$ relative to $a$.
            \label{lemma:upperbound:relative:regret}
        \end{lemma}
        \begin{proof}
            See  Appendix \ref{app:prf:lemma:upperbound:relative:regret}.
        \end{proof}
        Eq. \eqref{eqn:lemma:relative:regret}
        results to a simple upper-bound on the relative regret for small deviations between $\hat{a}$
        and $a$.
%        For small $(\hat{a}-a)^2 $, it
%        simplifies to
%        \begin{eqnarray}
%            \rregret{\hat{a}, a} \leq  {(\hat{a}-a)^2} \;  \fisher{X;a|Y}.
%            \label{eqn:upperbound:deviation:relative:regret}
%        \end{eqnarray}

        \subsection{Efficient Channel Estimation}

        Similar to the case for absolute regret, we now state the following
        result.
        \begin{lemma}
            Let $\mcal A$ denote the class of all unbiased estimation algorithms.
            If $\mcal A$ contains an {\em efficient estimator},
            the following holds true
            \begin{eqnarray}
                \inf_{A\in \mcal A}\Exp{\rregret{A_n(Y^{n-1}), a}} \leq  \frac{1}{n-1} \; \frac{ \fisher{X;a|Y}}{\fisher{Y;a}}
                \label{eqn:rregret:cramer-rao}
            \end{eqnarray}
            for sufficiently large values of $n$.
            \label{lemma:rregret:Cramer-Rao}
        \end{lemma}
        \begin{proof}
            The proof of this lemma is essentially the same as the proof of~Lemma~\ref{lemma:regret:Cramer-Rao}.
        \end{proof}
        Lemma \ref{lemma:rregret:Cramer-Rao} describes a bound on the expected relative regret, should
        an efficient estimator be used. This bound determines the smallest upper-bound on average relative regret, when
        sufficiently good unbiased channel estimators are used.

        \subsection{Regret Scalar}
        The constant value in the RHS of Eq. \eqref{eqn:rregret:cramer-rao} worths attention.
        It does not change with respect to $n$, and as $n \to \infty$, it becomes the sole
        scalar that determines the level of relative regret.
        We define this quantity as the {\em regret scalar} and denote it by
        \begin{eqnarray}
            \r(a)=\frac{\fisher{X;a|Y}}{\fisher{Y;a}}. % = \frac{\fisher{Y;a|X}}{\fisher{Y;a}}-1.
            \label{eqn:regret coefficient:fisher}
        \end{eqnarray}

        \begin{lemma}
            For every zero-mean input distribution $\probdist{}{X}$, the following trade-off holds true between
            regret scalar and output Fisher information
            \begin{eqnarray}
                (\r(a)+1) {\fisher{Y;a}} = \frac{\s^2_x}{\s^2_v}.
                \label{eqn:trade-off}
            \end{eqnarray}
            \label{lemma:tradeoff}
        \end{lemma}
        \begin{proof}
            See Appendix \ref{prf:lemma:tradeoff}.
        \end{proof}
        In Eq. \eqref{eqn:trade-off}, the RHS is the signal-to-noise ratio that is independent of $a$. Thus,
        Eq.~\eqref{eqn:trade-off} presents a simple product trade-off
        relationship between $\r(a)$ and $\fisher{Y;a}$. It suggest that the higher the Fisher information, the smaller the regret scalar,
        and vice-versa. The following example explicates this trade-off.

%
%        \begin{lemma}
%            Let $\dfi{a}{X}$ and $\dfi{b}{X}$ be probability density functions and
%            let
%             \[\kullback{\dfn{a}}{\dfn{b}} =\int \ln \frac{f_a}{f_b} f_a d \mu\]
%            denote the Kullback Leibler distance between them. The following
%            property holds true:
%            \begin{align}
%               \ddif{b} \kullback{\dfn{a}}{\dfn{b}} \bigg|_{b=a}= \fisher{X;a}.
%            \end{align}
%
%        \end{lemma}

    \begin{example}[Gaussian Input]
        Assume $X_n\sim\mcal N(0,\s^2_x)$ and
        $V_n\sim\mcal N(0,\s^2_v)$ are  i.i.d. implying that
        $Y_n\sim\mcal N(0,a^2\s^2_x+\s^2_v)$ and $Y_n|x_n \sim \mcal N(a x_n, \s^2_v )$.
        With perfect knowledge of $a$, the ideal estimator for $X$ given $Y=y$ is
        \begin{align}
            \phi_a(y) = \frac{a\s^2_x}{a^2\s^2_x+\s^2_v} y.
        \end{align}
        The MMSE error resulting from this estimator is
        \begin{align}
            {\Expi{}{(X-\phi_a(Y))^2}}=\frac{\s^2_x \s^2_v}{a^2 \s^2_x +\s^2_v}.%=\frac{\s^2_x} {\frac{a^2 \s^2_x}{\s^2_v}+1}.
        \end{align}
        A mismatch estimator for $\hat{a}$ is
        \begin{align}
            \phi_{\hat{a}}(y) = \frac{\hat{a}\s^2_x}{\hat{a}^2\s^2_x+\s^2_v} y.
        \end{align}
%        The absolute regret is
%        \begin{align}
%            \regret{\hat{a}, a} = \left( \frac{\hat{a}\s^2_x}{\hat{a}^2\s^2_x+\s^2_v}-\frac{a\s^2_x}{a^2\s^2_x+\s^2_v}\right)^2 (a^2\s^2_x+\s^2_v).
%        \end{align}
        We have
        \begin{eqnarray}
            \fisher{Y;a|X} = \frac{\s^2_x}{\s^2_v}
        \end{eqnarray}
         and
        \begin{align}
        \fisher{Y;a} =  \frac{2a^2\s^4_x}{(a^2\s^2_x+\s^2_v)^2}.
        \end{align}
%        Moreover, \[X|Y=y \sim \mcal N(\frac{a\s^2_x}{a^2\s^2_x +\s^2_v}y, \frac{\s^2_x \s^2_v}{a^2\s^2_x +\s^2_v} ).\]
%        Hence,
%        \begin{align}
%            \fisher{X;a|Y=y} =
%        \end{align}
        Thus,
        \begin{eqnarray}
            \r(a) = \frac{1}{2} \left(\frac{a^2\s^2_x}{\s^2_v} + \frac{\s^2_v}{a^2\s^2_x}\right)
        \end{eqnarray}
        Figure~\ref{fig:regret} depicts the behavior of $\r(a)$ and $\fisher{Y;a}$ with respect to
        $a$. The ${\rm SNR}=\frac{\s^2_x}{\s^2_v}=10 \text{ dB}$ and at $a=.35$, the minimum regret
        scalar coincides with maximum Fisher information.
            \begin{figure}[t]
                \begin{center}
                    \psfrag{S}[Bc][B1][\figfontsmall][0]{${\rm SNR} = \frac{\s^2_x}{\s^2_v}$}
                    \psfrag{Min}[Bc][B1][\figfontsmall][0]{Minimum Regret}
                    \psfrag{Max}[Bc][B1][\figfontsmall][0]{Maximum Fisher}
                    \psfrag{a}[Bc][B1][\figfontlarge][0]{$a$}
                    \includegraphics[width=.9\columnwidth]{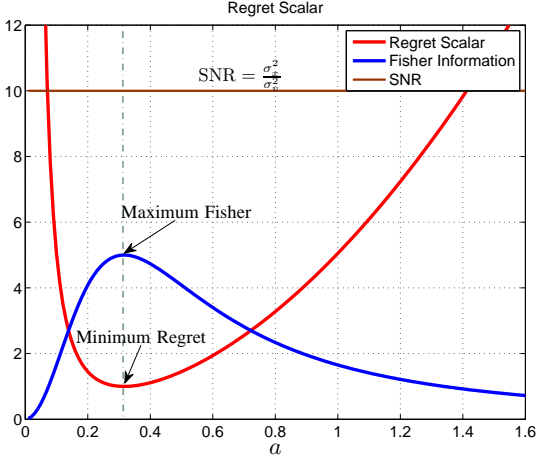}
                \end{center}
                \caption{Figure illustrates the multiplicative trade-off between regret scalar and Fisher information.
                Smaller Fisher information results to larger regret scalar and vice versa. The ${\rm SNR}=10 \text{ dB}$ and the minimum
                regret scalar is coincident with maximum Fisher information.}
                \label{fig:regret}
            \end{figure}

       % On the other hand, $\mu_x$ has some limited opposite effect. As $\mu_x$ increases,
       % $\r(a)$ reduces monotonically.
    \end{example}

    \section{Recap and Conclusion}
        We considered the problem of estimating the input signal  from the output of an additive white Gaussian
        noise  channel subject to parametric uncertainty. Namely, the channel gain is fixed, but unknown.
        In treating the problem, we considered an estimation process that consists of two algorithms: a blind channel
        estimator and a mismatched MMSE estimator to estimate the input.
        We studied the  regret that is raised as a result of mismatch estimation.
        Simple upper-bounds on both absolute and relative regrets were presented. These bounds
        provide useful tools in assessing deviation in estimating the
        input when there exists a small deviation in channel gain estimation.
        The bounds are simple and expressed in terms of the Fisher information. This makes them
        more intuitive and could potentially bridge to other known results in the literature.

        We also quantified regret for unbiased, { efficient} channel estimators. Using Caramer-Rao bound,
        we derived a simple trade-off between Fisher information and relative regret.
        This trade-off expresses that the  product of a certain function of relative
        regret and the Fisher information is equivalent to the signal-to-noise ratio, independent of the input distribution.
        The trade-off suggests that the higher the Fisher information, the smaller the expected relative regret.

        This work is our initial attempt to shed light on information-theoretic
        limits of blind deconvolution and dereverberation systems. We are currently working on generalization of
        these results to these applications.

\appendix

    \section{Proofs}

    \subsection{Proof of Lemma \ref{lemma:upperbound:regret}}
    \label{app:upperbound:regret}
        To derive an upperbound on absolute regret, we first state
        and prove the following results.
        \begin{proposition}
            For every $\hat{a}$ and $y \in \mcal Y$, we have
            \begin{align}
                \nonumber (\phi_{\hat{a}}(y) & - \phi_a(y))^2  \leq \\ & 2 (\Expi{\hat{a}}{X^2|y} + \Expi{{a}}{X^2|y})  \kullback{\probdistn{\hat{a}|y}}{\probdistn{{a}|y}}.
                \label{eqn:lemma:pointwise:regret}
            \end{align}
            \label{lemma:upperbound:pointwise:regret}
        \end{proposition}
        \begin{proof}
            By definition, we have
            \begin{align*}
                \left(\phi_{\hat{a}}(y)- \phi_a(y)\right)^2  &=  \left(\int x \big( \frac{d P_{\hat{a}|y} }{ d Q} - \frac{d P_{{a}|y} }{ d Q} \big) { d Q}  \right)^2
                %\label{eqn:pointwise:regret}
            \end{align*}
            for every probability measure $Q$ such that $P_{a|y} \ll Q$ and $P_{\hat{a}|y} \ll Q$.
            By Cauchy Schwartz inequality, we obtain
            \begin{align}
            \nonumber (\phi_{\hat{a}}(y)- \phi_a(y))^2  & \leq   \int x^2 \left(\sqrt{\frac{d P_{\hat{a}|y} }{ d Q}} + \sqrt{\frac{d P_{{a}|y} }{ d Q}} \right)^2 dQ \\
                        & .  \int  \left(\sqrt{\frac{d P_{\hat{a}|y} }{ d Q}} - \sqrt{\frac{d P_{{a}|y} }{ d Q}} \right)^2 dQ %\\
                \label{eqn:cauchy}
            \end{align}
            By inequality $(a+b)^2 \leq 2(a^2+b^2)$, one can show that the first term in the RHS of the above inequality is smaller than or equal to
            \[2 (\Expi{a}{X^2|y} + \Expi{\hat{a}}{X^2|y}).\]
            The second term in the RHS of inequality \eqref{eqn:cauchy}  is known as {\em Kakutani-Hellinger} distance  between $\probdist{{a}}{X|y}$ and $\probdist{\hat{a}}{X|y}$,
            denoted by \cite[p. 363]{Shiryaev:Probability}
            \begin{align*}
                r^2(\probdistn{\hat{a}|y}, \probdistn{{a}|y})=  \frac{1}{2}\int  \left(\sqrt{\frac{d P_{\hat{ a}|y} }{ d Q}} - \sqrt{\frac{d P_{{a}|y} }{ d Q}} \right)^2 dQ.
                %\label{eqn:kh:distance}
            \end{align*}
            Moreover, we know of the following inequality between Kakutani-Hellinger distance and Kullback-Leibler distance  \cite[p. 369]{Shiryaev:Probability}
            \begin{align*}
                2 r^2(\probdistn{\hat{a}|y}, \probdistn{{a}|y}) \leq \kullback{\probdistn{\hat{a}|y}}{\probdistn{{a}|y}}.
                %\label{eqn:kh:distance}
            \end{align*}
            Substituting in \eqref{eqn:cauchy}, we obtain Eq. \eqref{eqn:lemma:pointwise:regret}.
        \end{proof}

        \begin{proposition}
            For every ${a}$ and $y \in \mcal Y $, the following inequality holds true
            \begin{align}
                \Expi{{a}}{X^2|y}
                                    &\leq 3 \s^2_x + 4 \frac{y^2}{{a}^2}.
                \label{eqn:lemma:conditional:second:moment}
           \end{align}
            \label{lemma:upperbound:conditional:second:moment}
        \end{proposition}
        \begin{proof}
        Let $\dfi{{a}}{y|x}$ and $\dfi{{a}}{y}$ denote the conditional and marginal densities for
        $\probdist{{a}}{X,Y}$. Then,
        \begin{align}
            \nonumber \Expi{{a}}{X^2|y} &= \int x^2 \frac{\dfi{{a}}{y|x} }{\dfi{{a}}{y}} \dfi{}{x}dx\\
            \nonumber     & =\int_{x: \dfi{{a}}{y|x}\leq {\dfi{{a}}{y}}} + \int_{x: \dfi{{a}}{y|x}> {\dfi{{a}}{y}}}\\
                                &\leq \Expi{}{X^2} + \int_{x: \dfi{{a}}{y|x}> {\dfi{{a}}{y}}}
            \label{eqn:expectation:break}
        \end{align}
        To simplify the second term, we substitute $x^2$ by the inequality that is derived as follows
        \begin{align*}
            \dfi{{a}}{y|x} > &{\dfi{{a}}{y}} \Rightarrow\\
            (y-{a}x)^2 &< -2 \s^2_v \ln \left(\sqrt{2 \pi} \s_v \dfi{{a}}{y} \right). % \Rightarrow\\
%            \abs{y-\hat{a}x} &< \sqrt{-2 \s^2_v \ln \left(\sqrt{2 \pi} \s_v \dfi{\hat{a}}{y} \right)} \Rightarrow\\
%            \abs{\hat{a}x} &< \abs{y}+ \sqrt{-2 \s^2_v \ln \left(\sqrt{2 \pi} \s_v \dfi{\hat{a}}{y} \right)} \Rightarrow\\
%            x^2 &< 2\frac{y^2}{\hat{a}^2}-4 \frac{\s^2_v}{\hat{a}^2} \ln \left(\sqrt{2 \pi} \s_v \dfi{\hat{a}}{y} \right) \Rightarrow\\
%            x^2 &< 2\frac{y^2}{\hat{a}^2}+ 4 \frac{\s^2_v}{\hat{a}^2}  \int \frac{(y-\hat{a}x)^2}{2 \s^2_v} \dfi{}{x} dx  \Rightarrow\\
%            x^2 &< 4\frac{y^2}{\hat{a}^2}+ 2 \s^2_x
        \end{align*}
        Taking the square roots, we obtain
        \begin{align*}
 %           \dfi{\hat{a}}{y|x} > &{\dfi{\hat{a}}{y}} \Rightarrow\\
 %           (y-\hat{a}x)^2 &< -2 \s^2_v \ln \left(\sqrt{2 \pi} \s_v \dfi{\hat{a}}{y} \right) \Rightarrow\\
            \abs{y-{a}x} &< \sqrt{-2 \s^2_v \ln \left(\sqrt{2 \pi} \s_v \dfi{{a}}{y} \right)} \Rightarrow\\
            \abs{{a}x} &< \abs{y}+ \sqrt{-2 \s^2_v \ln \left(\sqrt{2 \pi} \s_v \dfi{{a}}{y} \right)}.% \Rightarrow\\
%            x^2 &< 2\frac{y^2}{\hat{a}^2}-4 \frac{\s^2_v}{\hat{a}^2} \ln \left(\sqrt{2 \pi} \s_v \dfi{\hat{a}}{y} \right) \Rightarrow\\
%            x^2 &< 2\frac{y^2}{\hat{a}^2}+ 4 \frac{\s^2_v}{\hat{a}^2}  \int \frac{(y-\hat{a}x)^2}{2 \s^2_v} \dfi{}{x} dx  \Rightarrow\\
%            x^2 &< 4\frac{y^2}{\hat{a}^2}+ 2 \s^2_x
        \end{align*}
        Taking the square of both sides of the previous inequality and using the inequality  $(a+b)^2 \leq 2(a^2+b^2)$, we obtain
        \begin{align*}
 %           \dfi{\hat{a}}{y|x} > &{\dfi{\hat{a}}{y}} \Rightarrow\\
%            (y-\hat{a}x)^2 &< -2 \s^2_v \ln \left(\sqrt{2 \pi} \s_v \dfi{\hat{a}}{y} \right) \Rightarrow\\
%            \abs{y-\hat{a}x} &< \sqrt{-2 \s^2_v \ln \left(\sqrt{2 \pi} \s_v \dfi{\hat{a}}{y} \right)} \Rightarrow\\
%            \abs{\hat{a}x} &< \abs{y}+ \sqrt{-2 \s^2_v \ln \left(\sqrt{2 \pi} \s_v \dfi{\hat{a}}{y} \right)} \Rightarrow\\
            x^2 &< 2\frac{y^2}{{a}^2}-4 \frac{\s^2_v}{{a}^2} \ln \left(\sqrt{2 \pi} \s_v \dfi{{a}}{y} \right) \Rightarrow\\
            x^2 &< 2\frac{y^2}{{a}^2}+ 4 \frac{\s^2_v}{{a}^2}  \int \frac{(y-{a}x)^2}{2 \s^2_v} \dfi{}{x} dx  \Rightarrow\\
            x^2 &< 4\frac{y^2}{{a}^2}+ 2 \s^2_x
        \end{align*}
        By substituting for $x^2$ in  the second term of the RHS of Eq. \eqref{eqn:expectation:break},
        we conclude Eq. \eqref{eqn:lemma:conditional:second:moment}.
        \end{proof}

        As a result of Propositions \ref{lemma:upperbound:pointwise:regret} and \ref{lemma:upperbound:conditional:second:moment}, we obtain
        \begin{align}
           \nonumber(\phi_{\hat{a}}(y)- \phi_a(y))^2 & \leq \\ & 2 (6 \s^2_x + 4 \frac{y^2}{\hat{a}^2}+4 \frac{y^2}{{a}^2}) \kullback{\probdistn{\hat{a}|y}}{\probdistn{{a}|y}}.
           \label{eqn:lemma:basic}
        \end{align}
        Moreover, the following equality is known between Kullback-Leibler distance and Fisher information \cite[p.55]{Kullback:InformationTheory}
        \begin{align}
            \nonumber \kullback{\probdistn{\hat{a}|y}}{\probdistn{{a}|y}} = \frac{(\hat{a}-a)^2}{2}& \fisher{X;a||Y=y}\\&+ o(\hat{a}-a)^2,
            \label{eqn:kulback:fisher}
        \end{align}
         where $ \fisher{X;a||Y=y}  \triangleq \Expi{}{\big(\nabla \ln \dfi{a}{X|Y}\big)^2|Y=y} $
        is the Fisher information of $X$ relative to $a$, conditioned on $Y=y$.
        Substitute Eq. \eqref{eqn:kulback:fisher} in Eq. \eqref{eqn:lemma:basic} and
        note that \[\frac{1}{\hat{a}^2}=\frac{1}{a^2}+o(\hat{a}-a)^2\]
        for $|\frac{\hat{a}-a}{a}|<<1$.
        Taking the expectation
        with respect to $Y$, we conclude the proof of Lemma \ref{lemma:upperbound:regret}.

    \subsection{Proof of Lemma \ref{lemma:regret:Cramer-Rao}}
        \label{app:prf:regret:Cramer-Rao}
        We know that $\hat{a}=A_n(Y^{n-1})$. For an unbiased estimator and for sufficiently large values of $n$, $|A_n(Y^{n-1})-a|<<1$ and
        \begin{align}
            \nonumber \regret{A_n(Y^{n-1}),a} & \leq (A_n(Y^{n-1})-a)^2 \\&\Exp{\left(6 \s^2_x + 8 \frac{Y^2}{{a}^2}\right) \fisher{X;a||Y}}.
            \label{eqn:upperbound:regret:small:deviation}
        \end{align}
        holds true with arbitrarily high probability. Taking the expectation
        of both sides of Eq. \eqref{eqn:upperbound:regret:small:deviation} with respect to $Y^{n-1}$, we obtain
        \begin{align*}
            \nonumber  & \Exp{\regret{A_n(Y^{n-1}), a}} \leq  \\
            &  \Exp{(A_n(Y^{n-1})-a)^2} \; \Exp{\left(6 \s^2_x + 8 \frac{Y^2}{{a}^2}\right)  \fisher{X;a||Y}}
            %\label{eqn:esterror:aregret}
        \end{align*}
        Take the infimum of both sides over $\mcal A$ and assume
        $\mcal A$ contains an {\em efficient estimator} \cite[p. 92]{Borovkov:MathematicalStatistics}.
        By definition an efficient estimator achieves the Cramer-Rao bound. This means
        \begin{align*}
            \Exp{(A_n(Y^{n-1})-a)^2} = \frac{1}{\fisher{Y^{n-1};a}}.
        \end{align*}
        Since $Y_n$ is i.i.d., by additivity of Fisher information \[\fisher{Y^{n-1};a} = (n-1) \fisher{Y;a}.\]
        As a result, we obtain
        \begin{align*}
            \nonumber \inf_{A\in \mcal A} \Exp{\regret{A_n(Y^{n-1}),  a}} \leq \qquad \qquad \qquad\\ \frac{1}{n-1}  \frac{\Exp{\left(6 \s^2_x + 8 \frac{Y^2}{{a}^2}\right)  \fisher{X;a||Y}}}{\fisher{Y;a}}.
        \end{align*}

    \subsection{Proof of Lemma \ref{lemma:upperbound:relative:regret}}
    \label{app:prf:lemma:upperbound:relative:regret}
        By Proposition \ref{lemma:upperbound:pointwise:regret}, we have
        \begin{align*}
           \frac{ (\phi_{\hat{a}}(y)- \phi_a(y))^2}{\Expi{\hat{a}}{X^2|y}+\Expi{a}{X^2|y}}   &\leq   2 \kullback{P_{\hat{a}|y}}{P_{a|y}}.
        \end{align*}
        Substituting from Eq. \eqref{eqn:kulback:fisher} and taking the average with respect to $Y$, we obtain
        \begin{align*}
             \rregret{\hat{a}, a}   \leq   (\hat{a}-a)^2  \Expi{}{\big(\nabla\ln \dfi{a}{X|Y}\big)^2}+ o(\hat{a}-a)^2,
        \end{align*}
        and conclude the proof.

   \subsection{Proof of Lemma \ref{lemma:tradeoff}}
        \label{prf:lemma:tradeoff}

        Since $X$ does not depend on $a$, $\fisher{X;a} = 0 $, and hence
        \begin{eqnarray*}
            \r(a)=\frac{\fisher{X;a|Y}}{\fisher{Y;a}} = \frac{\fisher{Y;a|X}}{\fisher{Y;a}}-1.
        \end{eqnarray*}
        Moreover, since the additive noise is Gaussian, the equality
        \begin{eqnarray*}
            \fisher{Y;a|X} = \frac{\s^2_x}{\s^2_v}
        \end{eqnarray*}
        holds true for every distribution $P(X)$ with zero mean. As a result, we
        obtain Eq. \eqref{eqn:trade-off}.

\vspace{10pt}
    \section*{Acknowledgement}
        The author is thankful to the anonymous reviewers whose
        comments and suggestions improved the presentation of the paper.

\end{document}

%% file: definitions.tex
\def\r{\rho}

\def\s{\sigma}

\def\f{\phi}

\def\mcal{\mathcal}

\DeclareMathAlphabet{\mathpzc}{OT1}{pzc}{m}{it}

\def\mbb{\mathbb} %comes with bbold package or with amsfont or amssymb
\newcommand{\abs}[1]{\lvert#1\rvert}

\renewcommand{\abs}[1]{\left| #1 \right|}
\newcommand{\Exp}[1]{{\mathbb E } \left [ #1 \right ]}
\newcommand{\Expi}[2]{{\mathbb E }_{#1} \left [ #2 \right ]}

\newcommand{\kullback}[2]{D(#1 \| #2)}

\newcommand{\probdistn}[1]{{P}_{#1}}%should be deleted
\newcommand{\probdist}[2]{{P}_{#1}(#2)}%should be deleted

%% file: environments.tex
\newtheorem{proposition}{Proposition}[{section}] %[{subsection}]
\newtheorem{lemma}{Lemma}[{section}] %[{subsection}]
\newtheorem{corollary}{Corollary}[{section}] %[{subsection}]
\newtheorem{example}{Example}[{section}] %

%% file: localmacro.tex
\newcommand{\regret}[1]{R(#1)}

\newcommand{\rregret}[1]{RR(#1)}

\newcommand{\dfi}[2]{f_{#1}(#2)}

\newcommand{\fisher}[1]{J(#1)}